%% file: unknown.tex
\documentclass[journal]{IEEEtran}

\usepackage[utf8]{inputenc}
\usepackage{hyperref}

\title{Permutations Unlabeled beyond Sampling Unknown}

\author{Ivan Dokmani\'c,~\IEEEmembership{Member,~IEEE}\thanks{Ivan Dokmani\'c is with the Coordinated Science Laboratory and the Department of Electrical Engineering at the University of Illinois at Urbana-Champaign.}\thanks{ID was supported by the National Science Foundation Award number 1817577, ``Combinatorial inverse problems in distance geometry'', and he wishes to thank Martin Vetterli for the title idea.}\thanks{This manuscript is a slightly revised version of a preprint available at \url{https://arxiv.org/abs/1812.00498}.}}

\include{commands}

\usepackage{mathabx}
\usepackage[noadjust]{cite}
\usepackage{mathtools}

\newcommand{\range}{\mathcal{R}}

\renewcommand{\ker}{\mathcal{N}}
\newcommand{\T}{\top}
\newcommand{\rref}{\mathsf{rref}}

\newcommand{\calC}{\mathcal{C}}
\newcommand{\Pit}{\mathit{\Pi}}

\begin{document}

\maketitle

\begin{abstract}
    A recent unlabeled sampling result by Unnikrishnan, Haghighatshoar and Vetterli  states that with probability one over iid Gaussian matrices $A$, any $x$ can be uniquely recovered from an unknown permutation of $y = A x$ as soon as $A$ has at least twice as many rows as columns. We show that this condition on $A$ implies something much stronger: that an unknown vector $x$ can be recovered from measurements $y = T A x$, when the unknown $T$ belongs to an arbitrary set of invertible, diagonalizable linear transformations $\mathcal{T}$. The set $\mathcal{T}$ can be finite or countably infinite. When it is the set of $m \times m$ permutation matrices, we have the classical unlabeled sampling problem.  We show that for almost all $A$ with at least twice as many rows as columns, all $x$ can be recovered either uniquely, or up to a scale depending on $\mathcal{T}$, and that the condition on the size of $A$ is necessary. Our proof is based on vector space geometry. Specializing to permutations we obtain a simplified proof of the uniqueness result of Unnikrishnan, Haghighatshoar and Vetterli. In this letter we are only concerned with uniqueness; stability and algorithms are left for future work.
\end{abstract}

\section{Introduction}

We ``steal'' a motivating example from \cite{Emiya:gi}: Imagine that you are recording a sound field with a large number of microphones connected to a recording interface. Alas, you forgot to label the cables so you end up with a pile of recordings without knowing which one corresponds to which spatial position. Is there a way to reconstruct the wavefield even without proper labels?

We can model this situation by the following unlabeled sampling problem:
\begin{equation}
    \label{eq:model}
    \tag{ULS}
    y = \Pit A x,
\end{equation}
where $A \in \C^{m \times n}$, $x \in \C^n$, and instead of measuring the usual $Ax$ we get to measure its unknown permutation. If the permutation $\Pit$ is known (the cables are neatly labeled), \eqref{eq:model} is simply a linear system. 

Many signal processing problems are modeled by \eqref{eq:model} and related constructions. If the columns of $A$ are samples of harmonic sinusoids, the problem is that of sampling at unknown locations \cite{Balakrishnan:1962,Marziliano:2000ki}. In simultaneous localization and mapping (SLAM), a robot is sensing an unknown environment, without ``knowing'' its own spatial location \cite{Krekovic:2016je}. If there are a finite number of possible locations and $A$ contains a model of the world as seen through the mobile sensors, then \eqref{eq:model} models a SLAM scenario.  A system similar to \eqref{eq:model} appears in room geometry reconstruction and microphone positionining by echoes \cite{Dokmanic:2013dz,Dokmanic:2014tc,Dokmanic:2016gu}. A nonlinear instance of sensing with unknown permutations is the unlabeled distance geometry problem where the task is to recover a point set from point-to-point distances, without knowing which distance corresponds to which pair of points \cite{Huang2018aa}.

In the context of Internet of things and fifth-generation communication systems, \eqref{eq:model} models header-free communication with very short packets \cite{Durisi:en,Song:ff}. Headers that identify individual nodes are too large compared with actual payloads, but in many sensing tasks the correct labeling can be inferred from the payload. When the nodes are sensors sensing a spatial field which has a subspace representation (for example, an advection-diffusion field \cite{Ranieri:2012wl,MartinezCamara2013hz} or a wavefield \cite{Ajdler:2006ex}), then the problem can be modeled as measuring $Ax$ up to a permutation. Recent work shows that the recovery can be addressed using symmetric polynomials \cite{Song:ff,tsakiris2018algebraic}.

Further connections exist with tomography with unknown projection angles, an especially relevant topic with the emergence of cryogenic electron microscopy (Cryo-EM) in which we get linear tomographic measurements with unknown angles \cite{Basu:2000bf,Zhao:2013cj}. Since the Radon transform has a restricted range, the problem can be modeled as \eqref{eq:model}

Problems of type \eqref{eq:model} can be split into underdetermined, $m < n$, and (over)determined ($m \geq n$). In the underdetermined case, we need a model for $x$. When $x$ is sparse, Emiya et al. \cite{Emiya:gi} adapt the branch-and-bound technique to efficiently search through all permutations.

We let $x$ be any complex vector and thus study the overdetermined case. In this setting, Unnikrishnan, Haghighatshoar, and Vetterli \cite{Unnikrishnan:2015gv,Unnikrishnan:2018gp} proved that if $A$ is iid Gaussian, it is possible to recover \emph{every} $x$ uniquely with probability 1 over realizations of $A$ if and only if $m \geq 2n$. Their proof involves sophisticated arguments from coding theory. Haghighatshoar and Caire also discuss recovery from an unknown but ordered subset of measurements \cite{Haghighatshoar:2017ks}. Pananjady, Wainwright, and Courtade discuss statistical and computational aspects of unlabeled linear regression \cite{Pananjady:2016vz}.

In this letter, we prove the following significant generalization of the above results: Imagine that $y$ was obtained as $T A x$ for some unknown invertible transformation $T \in \mathcal{T}$, where $\mathcal{T}$ is some set of invertible diagonalizable transformations, and $A$ is a known matrix. The set $\mathcal{T}$ can be finite or countably infinite. It can model unknown transfer functions, propagation parameters, and sensing parameters beyond permutations. We show that when $m \geq 2n$, for \emph{almost all} matrices $A$ \emph{all} $x$ can be recovered uniquely or up to a scale. Taking $\mathcal{T}$ to be the set of $m \times m$ permutation matrices (of cardinality $m!$), we recover the uniqueness result of Unnikrishnan, Haghighatshoar and Vetterli.

Our proof is simple and based on geometric arguments. The gist of it is that random $n$-dimensional subspaces of $\C^{2n}$ only intersect at the zero vector. On the other hand, if the ambient dimension is smaller than $2n$, then any two $n$-dimensional subspaces intersect non-trivially; we illustrate this in Figure \ref{fig:intersections}. The subtleties of the argument depend on the eigenvalues of transformations $T$; that it can be applied to permutation then follows by studying the eigenvalues of permutation matrices.

In this letter we are only concerned with the question of unique recovery. The important questions of recovery algorithms and their stability are left to future work.

Finally, after the first version of this preprint was published, Tsakiris \cite{tsakiris2018eigenspace} posted a preprint that also addresses the phenomenon we  described. Though \cite{tsakiris2018eigenspace} only deals with finite transformation classes, it extends our results to more general non-invertible transformations via algebraic-geometric arguments.

\vspace{-3mm}

\section{Main Result}

Our main lemma concerns the case of only two transformation matrices $\mathcal{T} = \set{I, T}$, where $I$ is the $m \times m$ identity matrix. We show that $x$ can be recovered from $y$ when $y$ is either $Ax$ or $TAx$, but we do not know which. The proof relies on studying the size of the intersection of the range of $A$ and the range of $TA$. 

We assume that $T \neq I$ and that $T \in \C^{m \times m}$ has an eigenvalue decomposition $T = \Phi \Lambda \Phi^{-1}$. We will denote by $\bar{\lambda}(T)$ an eigenvalue of $T$ with the largest multiplicity, and denote its multiplicity by $p(T)$. If there are multiple such eigenvalues, we break the tie arbitrarily as long as $1$ comes first. Note that the eigenvalues can be complex. Without loss of generality, we order the eigenvalues so that $\lambda_1 = \lambda_2 = \cdots = \lambda_p = \bar{\lambda}$. Since we assume that $T$ is diagonalizable, algebraic and geometric multiplicities coincide. 

We denote the Lebesgue measure on $\C^{m \times n}$ by $\mu$ and say that a property holds for almost all $A$ when it holds $\mu$-almost everywhere in $\C^{m \times n}$, that is, when it does not hold on $\mathcal{B} \subseteq \C^{m \times n}$ with $\mu(\mathcal{B}) = 0$. Since all the subsequent ``almost all'' claims in $\C^{m \times n}$ also hold almost everywhere in $\R^{m \times n}$ with respect to the Lebesgue measure on $\R^{m \times n}$, we extend the meaning of almost all to include both cases. We can then state the following:
\begin{lem}
    \label{lem:main}
    Let $T \in \C^{m \times m}$ be an invertible, diagonalizable matrix with an eigenvalue decomposition $T = \Phi \Lambda \Phi^{-1}$, $\Lambda = \diag(\lambda_1, \ldots, \lambda_m)$ and $A \in \C^{m \times n}$, $m \geq 2n$. Then for almost all matrices $A$, for all $y$ such that $y = Ax = TAz$, we have
    \begin{itemize}
        \item If $p(T) \leq m - n$, then $x = z$;
        \item If $p(T) > m - n$ and $\bar{\lambda}(T) = 1$, then $x = z$;
        \item If $p(T) > m - n$ and $\bar{\lambda}(T) \neq 1$, then $x = \bar{\lambda}(T) z$.
    \end{itemize}
\end{lem}

\begin{figure}
\includegraphics[width=\linewidth]{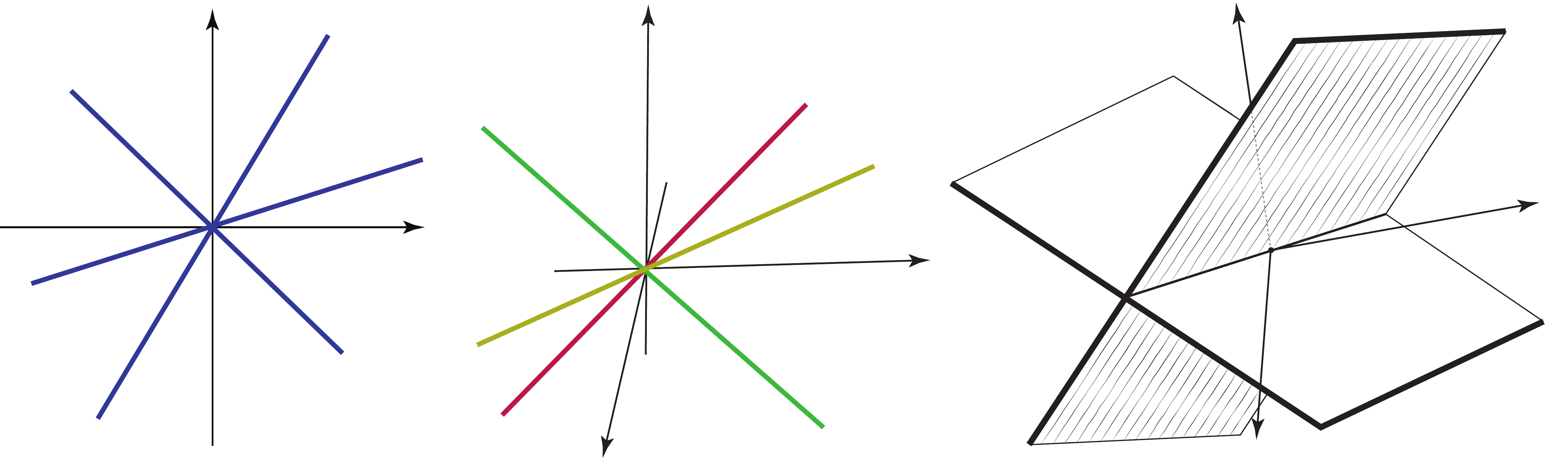}
\caption{Illustration of subspace intersections. Two 1D subspaces (lines through the origin) in 2D generically intersect only at the origin ($1 + 1 \leq 2$); the same holds for two 1D subspaces in 3D ($1 + 1 \leq 3$). Two 2D subspaces (planes through the origin) in 3D generically intersect along a line---a 1D subspace ($2 + 2 > 3$).}
\label{fig:intersections}
\vspace{-4mm}
\end{figure}

\begin{proof}
    We want to identify conditions on $A$ such that if 
    \begin{align}
    y &= A x \label{eq:nophi-model} \\
    \text{and} \quad y &= TAz = \Phi \Lambda \Phi^{-1} A z, \label{eq:phi-model}
    \end{align}
    we must have $x = z$. If \eqref{eq:nophi-model} and \eqref{eq:phi-model} hold simultaneously, then
    \[
        TAz \in \range(A).
    \]
    Thus, there exists a vector $y \in \range(A)$ such that also $Ty \in \range(A)$. Our proof hinges on the fact that this situation is very special.

    Write $A z = \Phi d$ for some $d$ (this is possible because the columns of $\Phi$ form a basis for $\C^m$). Then from \eqref{eq:phi-model} we have $y = \Phi \Lambda d$. Since $y \in \range(A)$, it must be that 
    \[
        \Lambda d = \Phi^{-1} y \in \range(\wt{A}),             
    \]
    where we defined the shortcut $\wt{A} = \Phi^{-1} A$. From the definition of $d$ we have $d = \Phi^{-1} A z$, so $d$ should also be in the range of $\wt{A}$. Another way to write this is as:
    \begin{align}
        \label{eq:nullspace-condition}
        \begin{rcases}
            d & \in \ \range(\wt{A}) \\
            \Lambda d & \in \ \range(\wt{A})
        \end{rcases}
        \ \ \Longleftrightarrow \ \ 
        \begin{cases}
            Q d & = \ 0 \\ 
            Q \Lambda d & = \  0
        \end{cases},
    \end{align}
    where $Q^*$ is the Hermitian transpose of $Q$, the columns of $Q^* \in \C^{m \times (m-n)}$ form a basis for the orthocomplement of the range of $\wt{A}$ and we used the fact that $\range(\wt{A}) = \ker(\wt{A}^*)^\perp$.

    Note that \eqref{eq:nullspace-condition} is a homogeneous system of $2 (m - n)$ equations in $m$ unknowns, so as soon as $2(m - n) < m$, that is to say, $m < 2n$, there are inevitably infinitely many solutions regardless of $\Lambda$. This case is further developed in Proposition \ref{prop:converse}. 

    Let $\rref$ denote the reduced row echelon form. For $A$ with full column rank (that is, for almost all $A$), $\wt{A}$ also has full column rank which implies
    \[
        \rref(\wt{A}^*) = \left[ \  I_{n \times n} \ | \ [S^*]_{n \times (m - n)} \ \right],
    \]
    with $S \in \C^{(m-n) \times n}$ (for convenience we indicate the block sizes in subscripts). From here we can read out a basis for $\ker(\wt{A}^*)$ as 
    \[
        Q^* =  
        \begin{bmatrix} 
            [S^*]_{n \times (m - n)} \\ - I_{(m - n) \times (m - n)}
        \end{bmatrix},
    \]
    with $S$ being full column  rank. Setting $\Lambda_1 = \diag(\lambda_1, \ldots, \lambda_n) \in \C^{n \times n}$, $\Lambda_2 = \diag(\lambda_{n+1}, \ldots, \lambda_{m})$, and partitioning $d$ as $d = [d_1^\T, \ d_2^\T]^\T$ we rewrite \eqref{eq:nullspace-condition} as
    \begin{align}
        S d_1 - d_2 &= 0, \\
        S \Lambda_1 d_1 - \Lambda_2 d_2 &= 0.
    \end{align}

    From the first equation we have $d_2 = S d_1$, so that
    \(
        S \Lambda_1 d_1 - \Lambda_2 S d_1 = 0,
    \)
    or
    \begin{equation}
        \label{eq:main-condition}
        (S \Lambda_1 - \Lambda_2 S) d_1 = 0.
    \end{equation}
    Let us focus on the top $n$ rows of this equation, with notation illustrated in Figure \ref{fig:dims}. In particular, let $\wt{S} \in \C^{n \times n}$ denote the top $n$ rows of $S$, and $\wt{\Lambda}_2 \in \C^{n \times n}$ the upper-left $n \times n$ block of $\Lambda_2$. From \eqref{eq:main-condition} we have that
    \begin{equation}
        \label{eq:main-condition-square}
        (\wt{S} \Lambda_1 - \wt{\Lambda}_2 \wt{S}) d_1 = 0.
    \end{equation}
    In order for a nonzero solution $d_1$ to exist, the system matrix must be singular:
    \begin{equation}
        \label{eq:determinant-equation}
        \det(\wt{S} \Lambda_1 - \wt{\Lambda}_2 \wt{S}) = 0.
    \end{equation}
    The determinant in \eqref{eq:determinant-equation} is a homogeneous multivariate polynomial in the entries of $\wt{S}$. This polynomial is either identically zero, or it is zero on a subset of $\C^{n \times n}$ of Lebesgue measure zero. 
    (Similarly, the set of real zeros has measure zero in $\R^{n \times n}$.) Consequently, it is either identically zero or nonzero for almost all $A$.

    \begin{figure}
    \includegraphics[width=\linewidth]{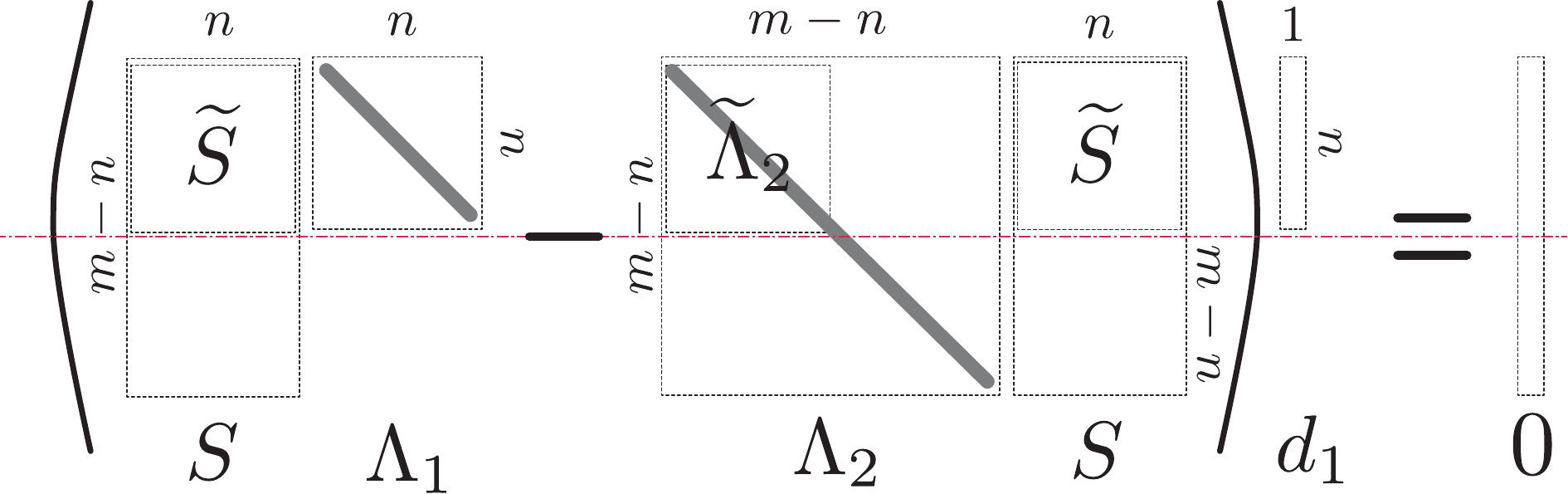}
    \caption{Dimensions of matrices in \eqref{eq:main-condition} and \eqref{eq:main-condition-square}.}
    \label{fig:dims}
    \vspace{-4mm}
    \end{figure}

    Setting $\wt{S} = I_n$ we get
    \(
        \det(\Lambda_1 - \wt{\Lambda}_2) = 0.
    \)
    Recall that $\Lambda_1 = \diag(\lambda_1, \ldots, \lambda_n)$, $\Lambda_2 = \diag(\lambda_{n+1}, \ldots, \lambda_{m})$, and that the eigenvalue with the largest multiplicity is listed first. 

    \begin{itemize}
        \item If $p \leq n$, clearly the determinant in \eqref{eq:determinant-equation} cannot be zero. Since the determinant is not identically zero, it is zero only for a set of $A$ of measure zero, hence for almost all $A$
        the only solution to \eqref{eq:main-condition-square} is $d_1 = 0$ which implies $y = 0$. Since almost all $A$ have full column rank, $x = z = 0$.
        \item If $n < p \leq m - n$ we can write \eqref{eq:main-condition} as 
        \begin{equation}
            \label{eq:p-larger-than-n}
            (S \Lambda_1 - \Lambda_2 S) d_1
            = 
            (\bar{\lambda} I - \Lambda_2) S d_1 = 0.
        \end{equation}
        The top $p - n$ rows of $(\bar{\lambda} I - \Lambda_2)$ are zero which implies the same for $(\bar{\lambda} I - \Lambda_2) S$, leaving us with $(m - n) - (p - n) = m - p \geq n$ independent nonzero equations ``at the bottom''. Using an analogous argument as above for the bottom $n$ equations, we again get that $d_1 = 0$ and $x = z$ for almost all $A$.
        \item If $m - n < p$, then $(\bar{\lambda} I - \Lambda_2) S$ has a nontrivial nullspace (we have fewer than $n$ nonzero equations for $d_1$), and \eqref{eq:p-larger-than-n} has a nonzero solution. 
    \end{itemize}    

    In the last case, any solution $d_1$ must be in the nullspace of the bottom $m - p$ rows of $S$ (this can be seen from \eqref{eq:p-larger-than-n} which also holds for $p > m - n$), so it must be that $d_2 = S d_1$ is supported only on the top $p - n$ entries (since $(m - p) + (p - n) = m - n)$. As a consequence, the vector $d = [d_1^\T, \ d_2^\T]^\T$ is supported on the top $n + (p - n) = p$ entries.

    Since $y = \Phi \Lambda d$, this implies that $y$ lives in the eigenspace spanned by the first $p$ eigenvectors corresponding to $\bar{\lambda}$. In summary, if $p > m - n$, then \emph{for all} $A$, $\range(A)$ and $\range(TA)$ intersect on the largest eigenspace corresponding to $\bar{\lambda}(T)$; \emph{for almost all} $A$ they do not intersect anywhere else.
    
    Thus for almost all $A$, all $s \in \range(A) \cap \range(TA)$ are such that $Ts = \bar{\lambda} s$. If $\bar{\lambda} = 1$, we can recover the corresponding $x$ uniquely since the equations $s = Ax$ and $s = TAz$ both have at most one solution, and the solution to the latter is $z = A^\dag T^{-1} s$, but $T^{-1} s = s$ and $A^\dag s = x$. Otherwise, if $\bar{\lambda} \neq 1$, we can recover up to a scaling since from $s = T A z$ we have $z = A^\dag T^{-1} s = \bar{\lambda}^{-1} A^\dag s = \bar{\lambda}^{-1} x$. 
\end{proof}




We now show how Lemma \ref{lem:main} implies a similar result for any number of unknown transformations in Theorem \ref{thm:main-result}. 

\begin{thm}
    \label{thm:main-result}
    Let $\mathcal{T} = \set{T_k \in \C^{m \times m}}_{k \in \mathcal{K}}$ be a finite or countably infinite set of invertible diagonalizable transforms, and $A \in \C^{m \times n}$, where $m \geq 2n$.
    Let further
    \(
        y = T A x,
    \)
    where $A \in \C^{m \times n}$, $x \in \C^n$, $T \in \mathcal{T}$, and neither $x$ nor $T$ are known. 
    Then 
    \begin{enumerate}
        \item If for all $T_1, T_2 \in \mathcal{T}$ we have that $p(T_1^{-1} T_2) \leq m - n$ or $\bar{\lambda}(T_1^{-1} T_2) = 1$ then  $x$ is uniquely determined by $y$.
        \item If there exist $T_1, T_2 \in \mathcal{T}$ such that $\bar{\lambda}(T_1^{-1} T_2) \neq 1$ and $p(T_1^{-1} T_2) > m - n$, then $x$ is determined up to a scale $\alpha \in \mathcal{A}$, where $\mathcal{A}$ is at most a countable set.
    \end{enumerate}
\end{thm}

\begin{proof}
    Denote by $\mathcal{T}_y$ the set of $T \in \mathcal{T}$ for which $y \in \range(T A)$. We only need to consider $T \in \mathcal{T}_y$. If $|\mathcal{T}_y| = 1$ we are done. 

    Suppose $|\mathcal{T}_y| > 1$ and let $T_1, T_2 \in \mathcal{T}_y$. That is, $y = T_1 A x = T_2 A z$ for some $x$ and $z$. Putting $\wt{y} = T_1^{-1} y$ we can write
    \begin{align}
        \label{eq:convert-to-identity}
        \begin{rcases}
            y & = \ T_1 Ax \\
            y & = \ T_2 Az
        \end{rcases}
        \ \ \Longleftrightarrow \ \ 
        \begin{cases}
            \wt{y} & = \ A x \\ 
            \wt{y} & = \ T_1^{-1} T_2 A z
        \end{cases}.
    \end{align}

    \begin{enumerate}
        \item By Lemma \ref{lem:main}, we have that for almost all matrices $A$, this implies $x = z$, for all $y$. In other words, the set $\mathcal{B}_{T_1, T_2}$ of ``bad'' matrices $A$ where it does not hold is of Lebesgue measure zero, $\mu(\mathcal{B}_{T_1, T_2}) = 0$. The set of matrices for which it might fail for any choice of $T_1$ and $T_2$ is
        \[
            \mathcal{B} = \bigcup_{\substack{T_1, T_2 \in \mathcal{T}_y\\T_1 \neq T_2}} \mathcal{B}_{T_1, T_2},
        \]
        but by the subadditivity of measure (and noting that the set of all pairs in $\mathcal{T}$ is countable),
        \[
            \mu(\mathcal{B}) \leq \sum_{\substack{T_1, T_2 \in \mathcal{T}_y\\T_1 \neq T_2}} \mu(\mathcal{B}_{T_1, T_2}) = 0.
        \]
        \item Again using Lemma \ref{lem:main} and reasoning as in 1), for a fixed $T_1, T_2$ and almost all $A$, we can uniquely recover any $x$ up to a scaling by $\bar{\lambda}(T_1^{-1} T_2)$. Thus the claim of the theorem holds with $\mathcal{A} = \set{\bar{\lambda}(T_1^{-1} T_2) \ : \ T_1, T_2 \in \mathcal{T}_y, T_1 \neq T_2}$.
    \end{enumerate}

\end{proof}

Theorem \ref{thm:main-result} establishes that under suitable conditions on $A$, for a rather general class of possible transformations $\mathcal{T}$, $x$ can be recovered from $y = TAx$, where $T \in \mathcal{T}$ is unknown. We now specialize these results to classical unlabeled sensing. We begin with a fact about  eigenvalues of permutations.

\begin{lem}
    \label{lem:permutation-ev}
    For any permutation matrix $\Pit$, $\bar{\lambda}(\Pit) = 1$.
\end{lem}

\begin{proof}
    Every $m \times m$ permutation $\Pit$ can be written as a product of $r$ disjoint cycles $\Pit = C_1 C_2 \cdots C_r$. Since the cycles are disjoint, the sum of the lengths is exactly $m$.

    Denote by $\mathcal{W}_i$ the set of $\ell_i$-th roots of unity, where $\ell_i$ is the length of the $i$th cycle $C_i$, $\mathcal{W}_i = \set{e^{\I 2\Pit p / \ell_i} \ : \ p \in \set{0, 1, \ldots, \ell_i - 1}}$. Then the set of all eigenvalues of $\Pit$ is \cite{Najnudel:2013bc}
    \[
        \mathcal{W} = \mathcal{W}_1 \cup \mathcal{W}_2 \cup \cdots \cup \mathcal{W}_r,
    \]
    and the geometric multiplicity of each eigenvalue is the number of times it appears in sets $\mathcal{W}_i$, $i \in \set{1, 2, \ldots, r}$. Note that every $\mathcal{W}_i$ contains a $1$, since $1 = e^{\I 2 \Pit 0}$. Therefore the eigenvalue of $\Pit$ with the largest multiplicity is $1$ (there could be other eigenvalues with the same maximal multiplicity).
\end{proof}

We will also use a partial converse to Lemma \ref{lem:main} to show that $m \geq 2n$ is necessary for permutations.

\begin{prop}[Partial converse]
    \label{prop:converse}
    Let $n \leq m < 2n$, and assume that there exist $T_1, T_2 \in \mathcal{T}$ such that $T_1^{-1} T_2$ has no eigenspace of dimension larger than $m - n$. Then for almost all matrices $A$, there exist $x, z \in \C^n$ such that $x \neq z$ and $T_1 A x = T_2 A z$. 
\end{prop}

\begin{proof}
    For almost all $A$ the ranges of $T_1 A$ and $T_2 A$ have dimension $n$. Since the sum of the dimensions of the two range spaces exceeds the dimension of the ambient space, $n + n > m$, they must have a non-trivial intersection:
    \begin{equation}
        \label{eq:intersection-exists}
        \dim \range(A) \cap \range(TA) \geq 2n - m \geq 1.
    \end{equation}
    Let $s \in S := \range(T_1 A) \cap \range(T_2 A)$. Then there exist $x$ and $z$ such that $s = T_1 Ax$ and $s = T_2 Az$.

    The unknown $x$ can be recovered only if $x = z$, in which case $A x$ is an eigenvector of $T_1^{-1} T_2$ with an eigenvalue 1 (recall equation \eqref{eq:convert-to-identity}). This can happen only if $\range(A)$ and the corresponding eigenspace of $T_1^{-1} T_2$ have a nontrivial intersection. But if every eigenspace $E_\lambda$ of $T_1^{-1} T_2$ has dimension at most $m - n$, then $\dim \range(A) + \dim(E_\lambda) \leq m$ and the two intersect only for a set of matrices $A$ of measure zero. Since from \eqref{eq:intersection-exists} a nonzero $s$ does exist, it must be that $x \neq z$.
\end{proof}

We can now easily prove the following:

\begin{cor}[Jayakrishnan, 2015]
    \label{cor:permutations}
    If $\mathcal{P} \in \R^{m \times m}$ is the set of all $m!$ permutation matrices of $m$ elements, then any $x$ can be uniquely recovered from measurements
    \(
        y = \Pit A x,
    \)
    where both $\Pit \in \mathcal{P}$ and $x \in \C^n$ are unknown, for almost matrices $A \in \C^{m \times n}$ with $m \geq 2n$. Conversely, if $m < 2n$ then for almost all $A$ there exist $x \neq z$ and permutations $\Pit_1 \neq \Pit_2$ such that $\Pit_1 A x = \Pit_2 A z$.
\end{cor}

\begin{proof}
Recoverability when $m \geq 2n$ is a straightforward consequence of Theorem \ref{thm:main-result} and Lemma \ref{lem:permutation-ev}, by noting that for any two permutations $\Pit_1$ and $\Pit_2$, $\Pit_1^{-1} \Pit_2$ is also a permutation. To prove the converse, note that cyclic shift by 1 (which is a permutation) has $m$ distinct eigenvalues $\set{\e^{\I 2k\Pit/m} \ : \ k \in \{0, 1, \ldots, m-1\}}$, so that all its eigenspaces have dimension 1. Whenever $m < 2n$, this implies $p(T) = 1 \leq m - n$, and the claim follows from Proposition \ref{prop:converse}. Namely, denoting the cyclic shift by $\Pit_c$, for all $\Pit_1, \Pit_2$ such that $\Pit_1^\T \Pit_2 = \Pit_c$ and almost all $A$, there exist $x, z$ such that $\Pit_1 A x = \Pit_2 A z$ and $x \neq z$.
\end{proof}

\vspace{-6mm}
\section{Extension to row-selection matrices} 
\label{sec:extension_to_row_selection_matrices}

In \cite{Unnikrishnan:2018gp} the authors state a more general result that allows row-selection matrices. They prove that instead of measuring a permutation of $Ax$, one can measure any (arbitrarily permuted) subset of $k$ entries of $Ax$ and still get unique recovery as long as $k \geq 2n$.

Though this case is beyond the scope of this letter, we outline an intuition for how our arguments might apply. The measurement can be written as $y = R \Pit A$, where $R$ is the top $k$ rows of an $m \times m$ identity matrix and $\Pit$ an unknown permutation. Same as before, it is sufficient to show that for almost all $A$ and two fixed permutations $\Pit_1$ and $\Pit_2$ 
\begin{equation}
    \label{eq:selection-equal}
    R \Pit_1 A x = R \Pit_2 A z
\end{equation}
implies $x = z$. Once that is established an argument parallel to that of Theorem \ref{thm:main-result} proves the result. 

Suppose that \eqref{eq:selection-equal} holds. Both $R \Pit_1 A$ and $R \Pit_2 A$ consist of $k$ rows of $A$ in some permuted order. Some rows of $A$, denote them by $A_C$, might be present in both $R \Pit_1 A$ and $R \Pit_2 A$, while some, denote them by $V, W$, appear in only one of them. We can thus rewrite \eqref{eq:selection-equal} as
\begin{equation}
    \label{eq:mismatched-rows}
    \Pit_1'
    \begin{bmatrix}
    V \\ A_{\calC}
    \end{bmatrix}
    x
    =
    \Pit_2'
    \begin{bmatrix}
    W \\ A_{C}
    \end{bmatrix}
    z
    \Longrightarrow
    \begin{bmatrix}
    V \\ A_{C}
    \end{bmatrix}
    x
    =
    \Pit
    \begin{bmatrix}
    W \\ A_{\calC}
    \end{bmatrix}
    z,
\end{equation}
for some permutations $\Pi_1'$ and $\Pi_2'$, where $\Pit = \Pit_1'^\T \Pit_2'$. We allow any of the blocks to be empty.

At one extreme where $A_C$ is empty, we can absorb $\Pit$ in $W$ and ask when it can be that $V x = W z$? But $\range(V) \cap \range(W) = \set{0}$ for almost all $V, W$ and hence almost all $A$ (see Figure~\ref{fig:intersections}). The detailed discussion of Lemma \ref{lem:main} is not needed because now $W$ varies independently of $V$ (the number of the degrees of freedom doubles).

At the other extreme where $V$ and $W$ are empty, $A_C$ has at least $2n$ rows. Lemma \ref{lem:main} and Corollary \ref{cor:permutations} guarantee that for almost all $A_C$ (and hence almost all $A$), the range of $A_C$ does not intersect the range of $\Pit A_C$ unless $\Pit$ has a large eigenspace, in which case this eigenspace corresponds to $\lambda = 1$. Interpolating between empty $V$ and $W$ and empty $A_C$, we are adding degrees of freedom and making the two matrices less dependent, which makes range intersections less likely.



\vspace{-2mm}
\section{Conclusion and future work} 
\label{sec:conclusion_and_future_work}

We presented a generalization of the classical unlabeled sampling canon. Instead of recovering $x$ from an unknown permutation of $y = Ax$, we showed that it can be recovered from rather general linear transformations of $y$ as long as the set of  transformations is at most countably infinite and $A$ has sufficiently many rows. As a byproduct, we get a simple, geometric proof of the uniqueness result for classical permutation-based unlabeled sensing.

The set of transformations $\mathcal{T}$ could model unknown room transfer functions where $A$ takes bandlimited spatial samples of speech. It could model different cameras and projections, or the variety of available sensors in any modality. In the classical unlabeled setting, we can expect the permutation ambiguity to be compounded by other uncertainties which can be modeled by $\mathcal{T}$---unknown filters, offsets, spatially-varying gains, etc.

An interesting line of future work is to relax assumptions on $T$. The fact that $T$ is diagonalizable or invertible does not seem essential, as long as its nullspace is not too large compared to the range space of $A$. It also seems plausible that nonlinear $T$ should work. The main practical question is that of stability and  polynomial-time recovery algorithms. For the case of permutations, results are beginning to emerge; these will point the way to algorithms for more general transformations.

\bibliographystyle{IEEEtran}
\bibliography{unknown}

\end{document}

%% file: commands.tex
\usepackage{amsmath}
\usepackage{amssymb}
\usepackage{amsthm}
\usepackage{latexsym}
\usepackage{verbatim}
\usepackage{graphicx}

\newtheorem{thm}{Theorem}
\newtheorem{prop}[thm]{Proposition}
\newtheorem{lem}[thm]{Lemma}
\newtheorem{cor}[thm]{Corollary}

{\begin{list}               
    {$\bullet$ \hfill}{
        \setlength{\leftmargin}{\parindent}
        \setlength{\parsep}{0.04\baselineskip}
        \setlength{\itemsep}{0.5\parsep}
        \setlength{\labelwidth}{\leftmargin}
        \setlength{\labelsep}{0em}}
    }
{\end{list}}

\providecommand{\cref}[1]{Chapter~\ref{chap:#1}}

\providecommand{\R}{\ensuremath{\mathbb{R}}}
\providecommand{\C}{\ensuremath{\mathbb{C}}}

\providecommand{\set}[1]{\left\{#1\right\}}

\providecommand{\diag}{\mathop{\mathrm{diag}}}

\providecommand{\e}{\ensuremath{\mathrm{e}}}
\providecommand{\I}{\ensuremath{\mathrm{i}}}

\providecommand{\wt}[1]{\ensuremath{\widetilde{#1}}}

